\documentclass{sig-alternate-05-2015}

\permission{\copyright 2017 International World Wide Web Conference Committee \\ (IW3C2), published under Creative Commons CC BY 4.0 License.}
\conferenceinfo{WWW 2017,}{April 3--7, 2017, Perth, Australia.}
\copyrightetc{ACM \the\acmcopyr}
\crdata{978-1-4503-4913-0/17/04. \\
http://dx.doi.org/10.1145/3038912.3052636 \\
\includegraphics{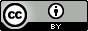}
}

\usepackage[ruled, noend, noline, linesnumbered]{algorithm2e}

\usepackage{caption}
\usepackage{subcaption}

\usepackage{paralist}
\usepackage{epstopdf}
\usepackage{color}
\usepackage{xcolor}
\usepackage{framed}
\usepackage{csvsimple}
\usepackage{longtable}
\usepackage[none]{hyphenat}
\usepackage{adjustbox}
\usepackage[pagebackref,letterpaper=true,colorlinks=true,pdfpagemode=none,urlcolor=blue,linkcolor=blue,citecolor=violet,pdfstartview=FitH]{hyperref}

\newtheorem{theorem}{Theorem}[section]

\newtheorem{claim}[theorem]{Claim}
\newtheorem{corollary}[theorem]{Corollary}
\newtheorem{definition}[theorem]{Definition}

\newcommand{\ignore}[1]{}

\newcommand{\cD}{\mathcal{D}}

\newcommand{\cS}{\mathcal{S}}
\newcommand{\cT}{{\cal T}}

\newcommand{\eps}{\varepsilon}

\newcommand{\bS}{\boldsymbol{S}}
\newcommand{\bT}{\boldsymbol{T}}

\newcommand{\NN}{\mathbb{N}}

\newcommand{\EX}{\hbox{\bf E}}

\newcommand{\sz}[1]{\mathrm{size}(#1)}

\newcommand{\sample}{{\tt sample}}
\newcommand{\shadow}{{\tt Shadow-Finder}}
\newcommand{\mainalg}{{\sc Tur\'{a}n-shadow}}

\newcommand{\Sec}[1]{\hyperref[sec:#1]{\S\ref*{sec:#1}}} 
\newcommand{\Eqn}[1]{\hyperref[eq:#1]{(\ref*{eq:#1})}} 
\newcommand{\Fig}[1]{\hyperref[fig:#1]{Fig.\,\ref*{fig:#1}}} 
\newcommand{\Tab}[1]{\hyperref[tab:#1]{Tab.\,\ref*{tab:#1}}} 
\newcommand{\Thm}[1]{\hyperref[thm:#1]{Theorem\,\ref*{thm:#1}}} 
\newcommand{\Fact}[1]{\hyperref[fact:#1]{Fact\,\ref*{fact:#1}}} 
\newcommand{\Lem}[1]{\hyperref[lem:#1]{Lemma\,\ref*{lem:#1}}} 
\newcommand{\Prop}[1]{\hyperref[prop:#1]{Prop.~\ref*{prop:#1}}} 
\newcommand{\Cor}[1]{\hyperref[cor:#1]{Corollary~\ref*{cor:#1}}} 
\newcommand{\Conj}[1]{\hyperref[conj:#1]{Conjecture~\ref*{conj:#1}}} 
\newcommand{\Def}[1]{\hyperref[def:#1]{Definition~\ref*{def:#1}}} 
\newcommand{\Alg}[1]{\hyperref[alg:#1]{Alg.~\ref*{alg:#1}}} 
\newcommand{\Ex}[1]{\hyperref[ex:#1]{Ex.~\ref*{ex:#1}}} 
\newcommand{\Clm}[1]{\hyperref[clm:#1]{Claim~\ref*{clm:#1}}} 
\newcommand{\Obs}[1]{\hyperref[obs:#1]{Obs.~\ref*{obs:#1}}} 
\newcommand{\Step}[1]{\hyperref[step:#1]{Step~\ref*{step:#1}}} 

\begin{document}
\sloppy




%

\title{A Fast and Provable Method for Estimating Clique Counts
Using Tur\'{a}n's Theorem}
%
%
%
%
%

\numberofauthors{2} 
%
\author{
%
%
\alignauthor
Shweta Jain \\
\affaddr{University of California, Santa Cruz}\\
\affaddr{Santa Cruz, CA}\\
\email{sjain12@ucsc.edu}
\alignauthor C. Seshadhri \\
\affaddr{University of California, Santa Cruz}\\
\affaddr{Santa Cruz, CA}\\
\email{sesh@ucsc.edu}
}

\maketitle

\begin{abstract}
Clique counts reveal important properties about the structure of massive graphs,
especially social networks.
The simple setting of just 3-cliques (triangles) has received much attention from 
the research community. For larger cliques (even, say $6$-cliques)
the problem quickly becomes intractable because of combinatorial explosion. Most methods used for triangle counting do not scale for large cliques, and existing algorithms require massive parallelism to be feasible. 

We present a new randomized algorithm that provably approximates the number of $k$-cliques,
for any constant $k$. The key insight is the use of (strengthenings of) the classic Tur\'an's theorem:
this claims that if the edge density of a graph is sufficiently high, the $k$-clique
density must be non-trivial. We define a combinatorial structure called a \emph{Tur\'an shadow},
the construction of which leads to fast algorithms for clique counting.

We design a practical heuristic, called \mainalg{}, based on this theoretical algorithm, and test 
it on a large class of test graphs. In all cases, \mainalg{} has less than 2\% error, and runs 
in a fraction of the time used by well-tuned exact algorithms. We do detailed comparisons
with a range of other sampling algorithms, and find that \mainalg{} is generally much faster
and more accurate. For example, \mainalg{} estimates all clique numbers up to size $10$ in social
network with over a hundred million edges. This is done in less than three hours on a single commodity
machine.

\end{abstract}

%
%


%
%

%
%



\begin{figure*}[t!]
    \begin{subfigure}[b]{0.3\textwidth}
    \centering
    \includegraphics[width=\textwidth]{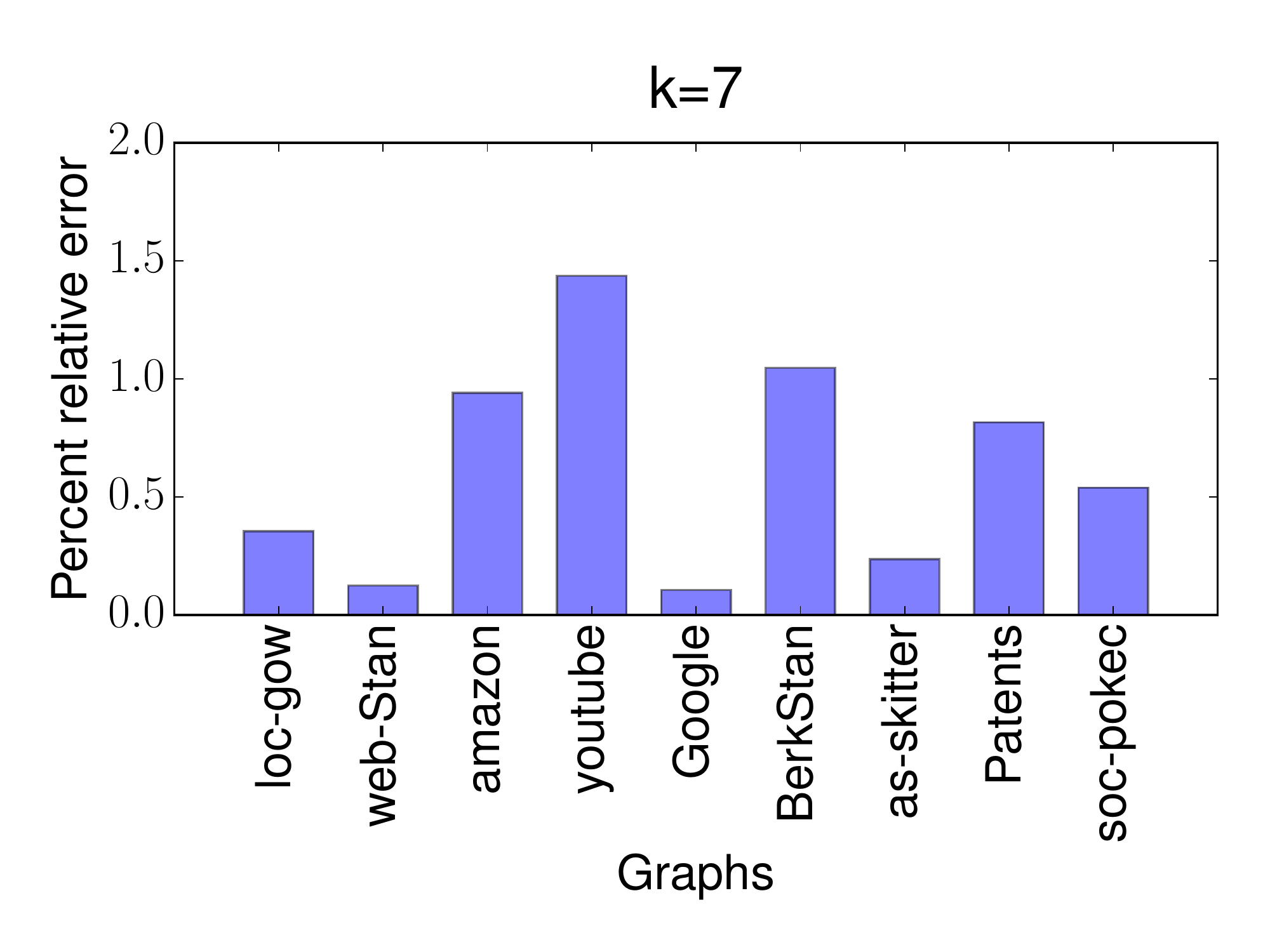}
        \caption{\footnotesize{Percent relative error for k=7}}
        \label{fig:acc}
    \end{subfigure}
~
    \begin{subfigure}[b]{0.3\textwidth}
    \centering
    \includegraphics[width=\textwidth]{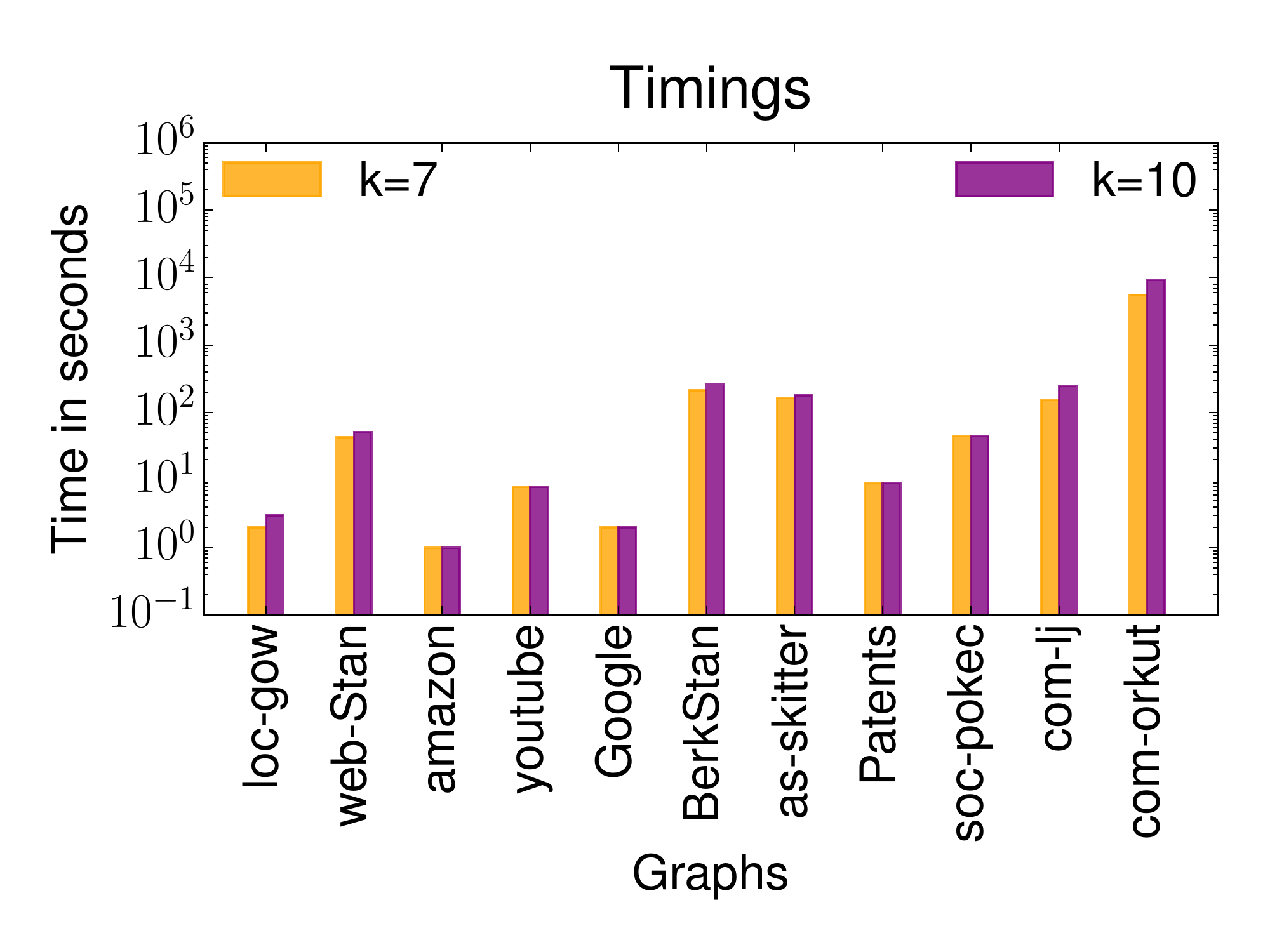}{}
        \caption{\footnotesize{Timings for k=7 and k=10}}
        \label{fig:timings}
    \end{subfigure}
~
    \begin{subfigure}[b]{0.3\textwidth}
    \centering
    \includegraphics[width=\textwidth]{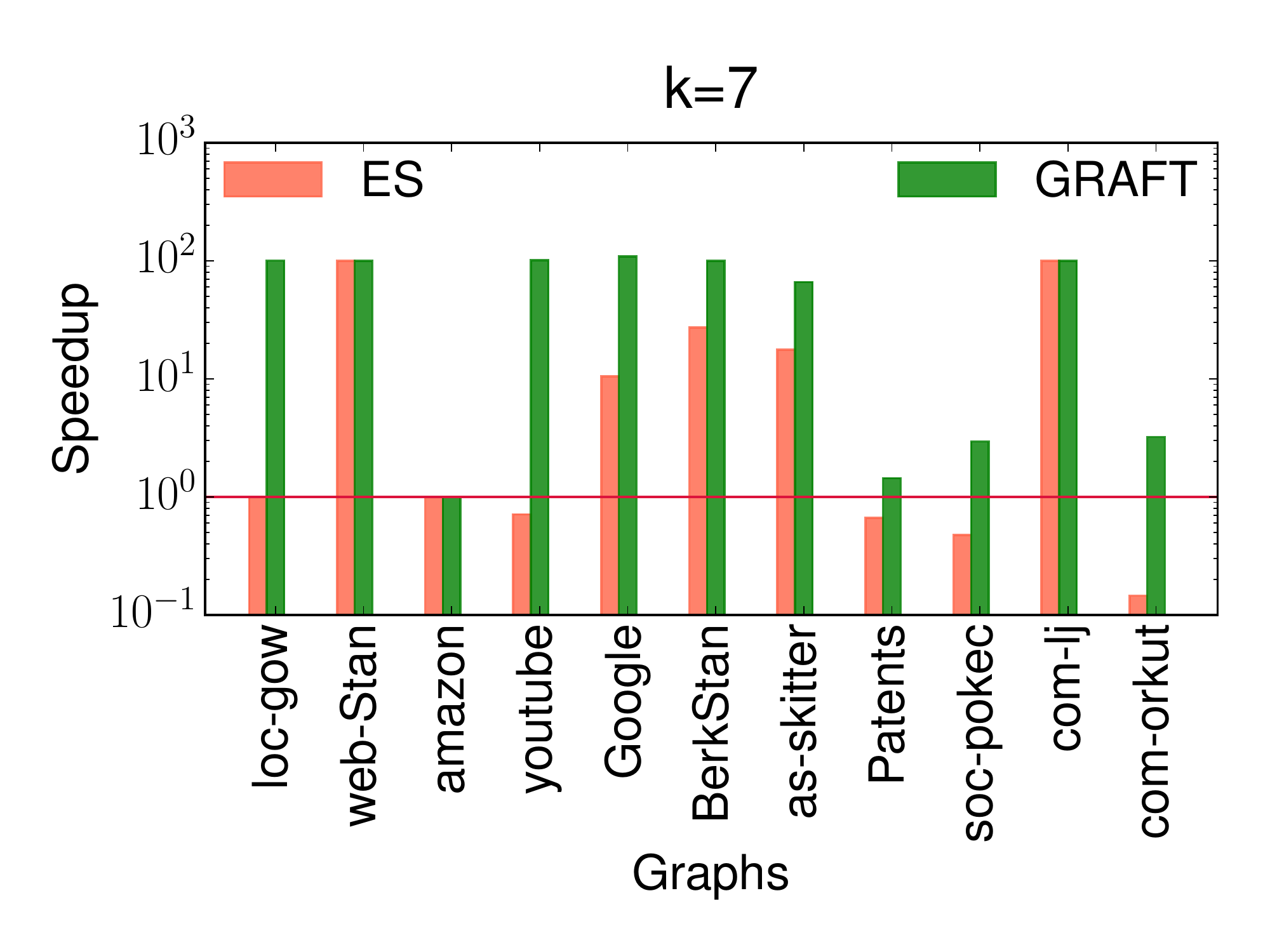}
        \caption{\footnotesize{Speedup for k=7}}
        \label{fig:7-speedup}
    \end{subfigure}
    \caption{ Summary of behavior of \mainalg{} over several datasets. \Fig{acc} shows the percent relative error in the estimates for k=7 given by \mainalg{}. We only show results for graphs for which we were able to obtain exact counts using either brute force enumeration, or from the results of~\cite{FFF15}. The errors are always $<2\%$ and mostly $<1\%$. \Fig{timings} shows the time taken by \mainalg{} for k=7 and k=10. \Fig{7-speedup} shows the speedup (time of algorithm/time of \mainalg) over other state of the art algorithms for k=7. The red line indicates a speedup of 1. We could not give a figure for speedup for k=10 because for most instances no competing algorithm terminated in min(7 hours, 100 times \mainalg{} time).}
\end{figure*}

\vfill\eject

\section{Introduction}\label{sec:intro}

Pattern counting is an important graph analysis tool in many domains:
anomaly detection, social network analysis, bioinformatics among others~\cite{HoLe70,Milo2002,Burt04,PrzuljCJ04,HoBe+07,Fa10}. 
Many real world graphs show significantly higher counts of certain patterns 
than one would expect in a random graph~\cite{HoLe70,WaSt98,Milo2002}. 
This technique has been referred to with a variety of names: subgraph analysis, motif counting,
graphlet analysis, etc. But the fundamental task is to count the occurrence
of a small pattern graph in a large input graph.
In all such applications, it is essential to have fast algorithms for pattern
counting.

It is well-known that certain patterns
capture specific semantic relationships, and thus the social dynamics are reflected
in these graph structures. The most famous such pattern
is the \emph{triangle}, which consists of three vertices connected to each other.
Triangle counting has a rich history in the social sciences and network
science~\cite{HoLe70,WaSt98,Burt04,FoDeCo10}. 

We focus on the more general problem of \emph{clique counting}. 
A \emph{$k$-clique} is a set of $k$ vertices that are all connected
to each other; thus, a triangle is a $3$-clique.
Cliques are extremely significant in social network analysis
(Chap. 11 of~\cite{HR05} and Chap. 2 of~\cite{J10}).
They are the archetypal example of a dense subgraph,
and a number of recent results use cliques to find large, dense
subregions of a network~\cite{SaSePi14,Ts15,MiPaPe+15,TsPaMi16}.

\subsection{Problem Statement} \label{sec:problem}

Given an undirected graph $G = (V,E)$, a $k$-clique is a set $S$ of $k$ vertices in $V$
with all pairs in $S$ connected by an edge. The problem is to count the number of $k$-cliques,
for varying values of $k$. Our aim is to get all clique
counts for $k \leq 10$.
    
The primary challenge is \emph{combinatorial explosion}. An autonomous system network
with ten million edges has more than a \emph{trillion} $10$-cliques. Any enumeration
procedure is doomed to failure. Under complexity theoretical assumptions, 
clique counting is believed to be exponential in the size $k$~\cite{CHKX04},
and we cannot hope to get a good worst-case algorithm.
Our aim is to employ \emph{randomized sampling} methods for clique counting,
which have seen some success in counting triangles and small
patterns~\cite{TsKaMiFa09,SePiKo13,JhSePi15}. 
We stress that we make no distributional assumption on the graph. 
All probabilities are over the internal randomness of the algorithm itself (which is independent of the instance).

%
%
%

\subsection{Main contributions} \label{sec:contri}

Our main theoretical result is a randomized algorithm \mainalg{} that approximates the
$k$-clique count, for any constant $k$. We implement this algorithm \emph{on a commodity machine}
and get $k$-clique counts (for all $k \leq 10$)
on a variety of data sets, the largest of which has 100M edges. The main features of our work follow.
    
\begin{asparaitem}
\item [\textbf{Extremal combinatorics meets sampling.}] Our novelty is in the
algorithmic use of classic extremal combinatorics results on clique densities. Seminal results of Tur\'{a}n~\cite{TURAN41} and Erd\H{o}s~\cite{E69} provide bounds on the number of cliques
in a sufficiently dense graph. 
\mainalg{} tries to cover $G$ by a carefully chosen collection of dense subgraphs that contains all cliques, called a \emph{Tur\'{a}n-shadow}. 
It then uses standard techniques to design an unbiased estimator for the clique count.
Crucially, the result of Erd\H{o}s~\cite{E69} (a quantitative version
of Tur\'{a}n's theorem) is used to bound the variance of the estimator.

We provide a detailed theoretical analysis of \mainalg, proving correctness and analyzing its time complexity. The running time of our algorithm is bounded by the time to construct the Tur\'{a}n-shadow,
which as we shall see, is quite feasible in all the experiments we run.
%
%

\item[\textbf{Extremely fast.}] In the worst case,
we cannot expect the Tur\'{a}n-shadow to be small, as that would
imply new theoretical bounds for clique counting. But in practice on a wide variety
of real graphs, we observe it to be much smaller than the worst-case bound.
Thus, \mainalg{} can be made into a \emph{practical} algorithm, which also
has provable bounds. We implement \mainalg{} and run it on a commodity machine.
\Fig{timings} shows the time required for \mainalg{} to obtain estimates for $k=7$ and $k=10$ 
in seconds. The {\tt as-skitter} graph is processed in less than 3 minutes,
despite there being billions of $7$-cliques and trillions of $10$-cliques.
All graphs are processed in minutes, except
for an Orkut social network with more than 100M edges (\mainalg{} handles
this graph within 2.5 hours).
\emph{To the best of our knowledge, there is no existing work that gets comparable results.}
An algorithm of Finocchi {\em et al.} also computes clique counts,
but employs MapReduce on the same datasets~\cite{FFF15}. We only require a single machine
to get a good approximation.

We tested \mainalg{} against a number of state of the art algorithmic
techniques (color coding~\cite{AlYuZw94}, edge sampling~\cite{TsKaMiFa09}, GRAFT~\cite{RaBhHa14}).
For 10-clique counting, none of these algorithms terminate for all instances even in 7 hours; \mainalg{} runs in minutes on all but one instance (where
it takes less than 2.5 hours).
For 7-clique counting, \mainalg{} is typically 10-100 times faster than competing 
algorithms.
 (A notable exception
is {\tt com-orkut}, where an edge sampling algorithm runs much faster than \mainalg.)

%
%

\item[\textbf{Excellent accuracy.}] \mainalg{} has extremely small variance,
and computes accurate results (in all instances we could verify).
We compute exact results for $7$-clique
numbers, and compare with the output of \mainalg. In~\Fig{acc}, we see that the
accuracy is well within 2\% (relative error)
of the true answer for all datasets. We do detailed experiments to measure
variance, and in all cases, \mainalg{} is accurate.

The efficiency and accuracy of \mainalg{} allows us to get clique
counts for a variety of graphs, and track how the counts change
as $k$ increases. We seem to get two categories of graphs:
those where the count increases (exponentially) with $k$,
and those where it decreases with $k$, see \Fig{trends}.  This provides a new lens
to view social networks, and we hope \mainalg{} can become
a new tool for pattern analysis.

\end{asparaitem}

\subsection{Related Work} \label{sec:related}

The importance of pattern counts gained attention in bioinformatics
with a seminal paper of Milo {\em et al.}~\cite{Milo2002},
though it has been studied for many decades in the social sciences~\cite{HoLe70}.
Triangle counting and its use has an incredibly rich history, and 
is used in applications as diverse as spam detection~\cite{BeBoCaGi08}, 
graph modeling~\cite{SeKoPi11}, and role detection~\cite{Burt04}.
Counting four cliques is mostly feasible using some recent developments
in sampling and exact algorithms~\cite{JhSePi15,AhNe+15}.

Clique counts are an important part of recent dense subgraph discovery 
algorithms~\cite{SaSePi14,Ts15}. Cliques also play an important role in 
understanding dynamics of social capital~\cite{JRT12}, and their
importance in the social sciences is well documented~\cite{HR05,J10}.
In topological approaches to network analysis, cliques are the fundamental
building blocks used to construct simplicial structures~\cite{SGB16}.

From an algorithmic perspective, clique counting has received
much attention from the theoretical computer science community~\cite{ChNi85,AlYuZw94,CHKX04,V09}.
Maximal clique enumeration has been an important topic~\cite{A73,Tomita04,ES11} since
the seminal algorithm of Bron-Kerbosch~\cite{BronKerb73}. Practical
algorithms for finding the maximum clique were given by Rossi {\em et al.} using
branch and bound methods~\cite{RossiG15}.

Most relevant to our work is a classic algorithm of Chiba and Nishizeki~\cite{ChNi85}.
This work introduces graph orientations to reduce the search time and provides
a theoretical connection to graph arboricity. We also apply this technique in \mainalg.

The closest result to our work is a recent MapReduce algorithm of Finocchi {\em et al.} 
for clique counting~\cite{FFF15}. This result applies the orientation technique of~\cite{ChNi85},
and creates a large set of small (directed) egonets. Clique counting overall 
reduces to clique counting in each of these egonets, and this can be parallelized using MapReduce.
We experiment on the same graphs used in~\cite{FFF15} (particularly, some of the largest ones)
and get accurate results on a single, commodity machine (as opposed to using a cluster).
Alternate MapReduce methods using multi-way joins have been proposed,
though this is theoretical and not tested on real data~\cite{AFRATI13}.

A number of randomized techniques have been proposed for pattern counting,
and can be used to design algorithms for clique counting. Most prominent
are color coding~\cite{AlYuZw94,HoBe+07,BetzlerBFKN11,ZhWaBu+12}
and edge sampling methods~\cite{TsKaMiFa09,TsKoMi11,RaBhHa14}. (MCMC methods~\cite{BhRaRa+12}
typically do not scale for graphs with millions of vertices~\cite{JhSePi15}.)
We perform detailed comparisons with these methods, and conclude that they
do not scale for larger clique counting.

\section{Main Ideas} \label{sec:ideas}

The starting point for our result is a seminal theorem of Tur\'{a}n~\cite{TURAN41}:
if the edge density of a graph is more than $1-\frac{1}{k-1}$, then
it must contain a $k$-clique. (The density bound is often called
the Tur\'{a}n density for $k$.) Erd\H{o}s proved a stronger version~\cite{E69}.
Suppose the graph has $n$ vertices. Then in this case, it contains
$\Omega(n^{k-2})$ $k$-cliques!

Consider the trivial randomized algorithm to estimate $k$-cliques.
Simply sample a uniform random set of $k$ vertices and check
if they form a clique. Denote the number of $k$-cliques by $C$, then
the success probability is $C/{n\choose k}$. Thus,
we can estimate this probability using ${n\choose k}/C$ samples.
By Erd\H{o}s' bound, $C = \Omega(n^{k-2})$.
Thus, if the density of a graph (with $n$ vertices) is above the Tur\'{a}n density, one can estimate
the number of $k$-cliques using $O(n^2)$ samples.

Of course, the input graph $G$ is unlikely to have such a high density,
and $O(n^2)$ is a large bound. We try to cover all $k$-cliques in $G$
using a collection of dense subgraphs. This collection is called
a \emph{Tur\'{a}n shadow}. We employ orientation techniques from Chiba-Nishizeki
to recursively construct a shadow~\cite{ChNi85}.

We take the degeneracy (k-core) ordering in $G$~\cite{SEIDMAN83}. It is well-known
that outdegrees are typically small in this ordering. To count $k$-cliques
in $G$, it suffices to count $(k-1)$-cliques in every outneighborhood.
(This is the main idea in the MapReduce algorithms of Finocchi et al~\cite{FFF15}.)
If an outneighborhood has density higher than the Tur\'{a}n density for $(k-1)$,
we add this set/induced subgraph to the Tur\'{a}n shadow. If not, we recursively
employ this scheme to find denser sets.

When the process terminates, we have a collection of sets (or induced subgraphs)
such that each has density above the Tur\'{a}n threshold (for some appropriate
$k'$ for each set). Furthermore, the sum of cliques ($k'$-cliques, for the same $k'$)
is the number of $k$-cliques in $G$. Now, we can hope to use
the randomized procedure to estimate the number of $k'$-cliques in each set
of the Tur\'{a}n shadow. By a theorem of Chiba-Nishizeki~\cite{ChNi85}, we can argue that number
of vertices in any set of the Tur\'{a}n shadow is at most $\sqrt{2m}$ 
(where $m$ is the number of edges in $G$). Thus, $O(m)$ samples suffices
to estimate clique counts for any set in the Tur\'{a}n shadow.

But the Tur\'{a}n shadow has many sets, and it is infeasible to spend $O(m)$ samples
for each set. We employ a randomized trick. We only need to approximate
the sum of clique counts over the shadow, and can use random sampling for that purpose.
Working through the math, we effectively set up a distribution over the sets
in the Tur\'{a}n shadow. We pick a set from this distribution, pick some subset of random vertices,
and check if they form a clique. The probability of this event can be related
to the number of $k$-cliques in $G$. Furthermore, we can prove that $O(m)$
samples suffice to estimate this probability. All in all, after constructing the Tur\'{a}n shadow,
$k$-clique counting can be done in $O(m)$ time.

\subsection{Main theorem and significance} \label{sec:thm}

The formal version of the main theorem
is \Thm{main-full}. It requires a fair bit of terminology
to state. So we 
state an informal version that maintains
the spirit of our main result. This should provide the reader
with a sense of what we can hope to prove.
We will define the Tur\'{a}n shadow formally
in later sections. But it basically refers to the construct described
above.

\begin{theorem} \label{thm:main} [Informal] Consider graph $G=(V,E)$ with $n$ vertices,
$m$ edges, and maximum core number $\alpha$. 
Let $\bS$ be the Tur\'{a}n $k$-clique shadow of $G$, and let $|\bS|$
be the number of sets in $\bS$.

Given any $\delta > 0, \eps > 0, k$, with probability at least $1-\delta$, the procedure \mainalg{} outputs a $(1+\eps)$-multiplicative
approximation to the number of $k$-cliques in $G$. The
running time is linear in $|\bS|$ and $m \alpha \log(1/\delta)/\eps^2$.
The storage is linear in $|\bS|$.
\end{theorem}

Observe that the size of the shadow is critical to the procedure's efficiency. As long
as the number of sets in the Tur\'{a}n shadow is small,
the extra running time overhead is only linear in $m$. And in practice,
we observe that the Tur\'{a}n shadow scales linearly with graph size, leading
to a practically viable algorithm.

\medskip

\textbf{Outline:} In \Sec{prelim}, we formally describe Tur\'{a}n's theorem and
set some terminology. \Sec{shadows} defines (saturated) shadows, and shows how to 
construct efficient sampling algorithms for clique counting from shadow.
\Sec{const} describes the recursive construction of the Tur\'{a}n shadow.
In \Sec{together}, we describe the final procedure \mainalg{}, and prove
(the formal version of) \Thm{main}. Finally, in \Sec{results}, we
detail our empirical study of \mainalg{} and comparison with the state of the art.

\section{Tur\'{a}n's Theorem} \label{sec:prelim}

For any arbitrary graph $H = (V(H), E(H))$, let $C_i(H)$ denote the set of cliques
in $H$, and $\rho_i(H) := |C_i(H)|/{|V(H)| \choose i}$
is the $i$-clique density. Note that $\rho_2(H)$ is the standard
notion of edge density. 

The following theorem of Tur\'{a}n is one of the most important
results in extremal graph theory.

\begin{theorem} \label{thm:turan} (Tur\'an~\cite{TURAN41}) For any graph $H$, 
if $\rho_2(H) > 1-\frac{1}{k-1}$, then $H$ contains a $k$-clique.
\end{theorem} 

This is tight, as evidenced by the complete $(k-1)$-partite graph
$T_{n,k-1}$ (also called the Tur\'{a}n graph). In a remarkable generalization,
Erd\H{o}s proved that if an $n$-vertex graph has even \emph{one more edge}
than $T_{n,k-1}$, it must contain many $k$-cliques. One can think
of this theorem as a quantified version of Tur\'{a}n's theorem.

\begin{theorem} \label{thm:erdos} (Erd\H{o}s~\cite{E69}) For any graph $H$ over $n$ vertices,
 if $\rho_2(H)>1-\frac{1}{k-1}$, then $H$ contains at least $(n/(k-1))^{k-2}$ $k$-cliques.
\end{theorem} 

It will be convenient to express this result in terms on $k$-clique densities.
We introduce some notation: let $f(k) = k^{k-2}/k!$. By Stirling's approximation,
$f(k)$ is well approximated by $e^k/\sqrt{2\pi k^{5}}$. Note that $f(k)$
is some fixed constant, for constant $k$.
This corollary will be critical to our analysis.

\begin{corollary} \label{cor:erdos} For any graph $H$ over $n$ vertices, if $\rho_2(H)>1-\frac{1}{k-1}$, 
then $\rho_k(H) \geq 1/f(k)n^2 $.
\end{corollary}

\begin{proof} 
By \Thm{erdos}, $H$ has at least $(\frac{n}{(k-1)})^{k-2}$ $k$-cliques. Thus,
\begin{equation*}
\rho_k(H) \geq \frac{(\frac{n}{(k-1)})^{k-2}}{{n \choose k}} \geq n^{k-2}/n^k \times k!/(k-1)^{k-2} \geq 1/(f(k) n^2)
\end{equation*}
\end{proof} 

%
%
%
%

\section{Clique shadows} \label{sec:shadows}

A key concept in our algorithm is that of \emph{clique shadows}. 
Consider graph $G = (V,E)$. For any set $S \subseteq V$, we let $C_\ell(S)$
denote the set of $\ell$-cliques contained in $S$.

\begin{definition} \label{def:shadow} A $k$-clique shadow $\bS$ for graph $G$ 
is a multiset of tuples $\{(S_i, \ell_i)\}$ where $S_i \subseteq V$ and $\ell_i \in \NN$
such that: there is a bijection between $C_k(G)$ and $\bigcup_{(S,\ell) \in \bS} C_\ell(S)$.

Furthermore, a $k$-clique shadow $\bS$ is $\gamma$-saturated if $\forall (S,\ell) \in \bS$,
$\rho_\ell(S) \geq \gamma$.
\end{definition}

Intuitively, it is a collection of subgraphs, such that the sum of clique counts within them
is the total clique count of $G$. Note that for each set $S$ in the shadow, the 
associated clique size $\ell$ is different (for different $S$).
Observe that $\{(V,k)\}$ is trivally a clique shadow. But it is highly unlikely to be saturated.

It is important to define the \emph{size} of $\bS$, which is really the storage required
to represent it.

\begin{definition} The representation size of $\bS$ is denoted $\sz{\bS}$,
and is $\sum_{(S,\ell) \in \bS} |S|$.
\end{definition}

\begin{algorithm}
\caption{\sample$(\bS,\gamma,k,\eps,\delta)$ \newline 
$\bS$ is $\gamma$-saturated $k$-clique shadow\newline
$\eps, \delta$ are error parameters }
 For each $(S,\ell) \in \bS$, set $w(S) = {|S|\choose \ell}$ \;
 Set probability distribution $\cD$ over $\bS$ where $p(S) = w(S)/\sum_{(S,\ell) \in \bS} w(S)$ \;
 \label{step:sample} For $r \in 1, 2, \ldots, t = \frac{20}{\gamma\eps^2} \log(1/\delta)$\;
 \ \ \ \ Independently sample $(S,\ell)$ from $\cD$\;
 \ \ \ \ Choose a u.a.r. $\ell$-tuple $A$ from $S$\;
 \ \ \ \ \label{step:X} If $A$ forms $\ell$-clique, set indicator $X_r = 1$. Else, $X_r = 0$ \;
 Output $\frac{\sum_r X_r}{t} \sum_{(S,\ell) \in \bS} {|S| \choose \ell}$ as estimate for $|C_k(G)|$\;
 \end{algorithm}

When a $k$-clique shadow $\bS$ is $\gamma$-saturated, each $(S,\ell) \in \bS$ has many $\ell$-cliques.
Thus, one can employ random sampling within each $S$ to estimate $|C_\ell(S)|$,
and thereby estimate $C_k(G)$. We use a sampling trick to show that we do not
need to estimate all $|C_\ell(S)|$; instead we only need $O(1/\gamma)$ samples
in total. 

 \begin{theorem} \label{thm:sample} Suppose $\bS$ is a $\gamma$-saturated $k$-clique shadow
 for $G$. The procedure \sample$(\bS)$ outputs an estimate $\hat{C}$ such $|\hat{C} - |C_k(G)|| \leq \eps |C_k(G)|$ with probability $> 1- \delta$.

 The running time of \sample$(\bS)$ is $O(\sz{\bS} + \frac{1}{\gamma\eps^2}\log(1/\delta))$.
 \end{theorem}

 \begin{proof} We remind the reader that $w(S) = {|S|\choose \ell}$.
 Set $\alpha = |C_k(G)|/\sum_{S \in \bS} w(S)$. Observe that 
 \begin{eqnarray*}
 \Pr[X_r = 1] & = & \sum_{(S,\ell) \in \bS} \Pr[\textrm{$(S,\ell)$ is chosen}] \\
 & \times & \Pr[\textrm{$\ell$-clique chosen in $S$} | \textrm{$(S,\ell)$ is chosen}] 
 \end{eqnarray*}
 The former probability is exactly $w(S)/\sum_{S \in \bS} w(S)$, and the 
 latter is exactly $|C_\ell(S)|/{|S| \choose \ell}$ $=|C_\ell(S)|/w(S)$.
 So, 
 \begin{equation*}
 \Pr[X_r = 1] = \sum_{(S,\ell) \in \bS} |C_\ell(S)|/\sum_{S \in \bS} w(S)
 \end{equation*}
 Since $\bS$ is a $k$-clique shadow, $\sum_{(S,\ell) \in \bS} |C_\ell(S)| = |C_k(G)|$.
 Thus, $\Pr[X_r = 1] = \alpha$.
 By the saturation property, $\rho_\ell(S) \geq \gamma$, equivalent to 
 $|C_\ell(S)| \geq \gamma w(S)$. So $\sum_{S \in \bS} |C_\ell(S)| \geq
 \gamma \sum_{S \in \bS} w(S)$. That implies that $\alpha \geq \gamma$.
 By linearity of expectation, $\EX[\sum_{r \leq t} X_r] = \sum_{r\leq t} \EX[X_r] \geq \gamma t$.

 Note that all the $X_r$s come from independent trials. (The graph structure plays no role,
 since the distribution of each $X_r$ does not change upon conditioning on the other $X_r$s.)
 By a multiplicative Chernoff bound (Thm 1.1 of~\cite{DuPa09}), 
 \begin{eqnarray*}
 & & \Pr[\sum_r X_r/t \leq \alpha(1-\eps)] \leq \exp(-\eps^2 \EX[\sum_r X_r]/3) \\
 & \leq & \exp(-\eps^2 \gamma t/3) = \exp(-5\log(1/\delta)) \leq \delta/5.
 \end{eqnarray*}
 By an analogous upper tail bound, $\Pr[\sum_r X_r/t \geq \alpha(1+\eps)] \leq \delta/5$.
 By the union bound, with probability at least $1-2\delta/5$, 
 $\alpha(1-\eps) \leq \sum_r X_r/t \leq \alpha(1+\eps)$.
 Note that the output $\hat{C} = (\sum_r X_r/t) \sum_{S \in \bS} w(S)$.
 We multiply the bound above on $\sum_r X_r/t$ by $\sum_{S \in \bS} w(S)$,
 and note that $\alpha \sum_{S \in \bS} w(S) = |C_k(G)|$ to complete the proof.
 \end{proof}

 We stress the significance of \Thm{sample}. Once we get a $\gamma$-saturated clique shadow $\bS$,
 $|C_k(G)|$ can be approximated in time \emph{linear} in $\sz{\bS}$. The number of samples
 chosen only depends on $\gamma$ and the approximation parameters, not on the graph size. 

 But how to actually generate a saturated clique shadow? Saturation
 appears to be extremely difficult to enforce. This is where the theorem of Erd\H{o}s (\Thm{erdos})
 saves the day. It merely suffices to make the edge density of each set in the clique shadow
 high enough. The $k$-clique density \emph{automatically} becomes large enough.

 \begin{theorem} \label{thm:sat-shadow} Consider a $k$-clique shadow $\bS$ 
 such that $\forall (S,\ell) \in \bS$, $\rho_2(S) > 1-\frac{1}{\ell-1}$.
 Let $\gamma = 1/\max_{(S,\ell)\in \bS} f(\ell)|S|^2$.
 Then, $\bS$ is $\gamma$-saturated.
 \end{theorem}

 \begin{proof} By \Cor{erdos}, for every $(S,\ell) \in \bS$,
 $\rho_\ell(S) \geq 1/(f(\ell) |S|^2)$. We simply set $\gamma$ to
 be the minimum such density over all $(S,\ell) \in \bS$.
\end{proof}

\section{Constructing saturated clique shadows} \label{sec:const}

We use a refinement process to construct saturated clique shadows.
We start with the trivial shadow $\bS = \{(V,k)\}$ and iteratively
``refine" it until the saturation property is satisfied. By \Thm{sat-shadow},
we just have to ensure edge densities in each set are sufficiently large.

For any set $S \subset V$, let $G|_S$ be the subgraph of $G$ induced by $S$. Given an unsaturated $k$-clique shadow $\bS$, we find some $(S,\ell) \in \bS$
such that $\rho_2(S) \leq 1-\frac{1}{\ell-1}$. By iterating over the vertices,
we replace $(S,\ell)$ by various neighborhoods in $G|_S$ to get a new shadow.
We would like the edge densities of these neighborhoods to increase,
in the hope of crossing the threshold given in \Thm{sat-shadow}. 

The key insight is to use the \emph{degeneracy ordering} to construct
specific neighborhoods of high density that also yield a valid shadow.
This is basically the classic graph theoretic technique of computing
core decompositions, which is widely used in large-graph analysis~\cite{SEIDMAN83, GiChMa14}.
As mentioned earlier, this idea is used for fast clique counting as well~\cite{ChNi85,FFF15}.

\begin{definition} \label{def:degen} For a (labeled) graph $G = (V,E)$, a
\emph{degeneracy ordering} is a permutation of $V$ given as $v_1, v_2, \ldots, v_n$
such that: for each $i \leq n$, $v_i$ is the minimum degree
vertex in the subgraph induced by $v_i, v_{i+1}, \ldots, v_n$.
(As defined, this ordering is not unique, but we can enforce
uniqueness by breaking ties by vertex id.)

The degree of $v_i$ in $G|_{\{v_i,\ldots,v_n\}}$ is the core number
of $v_i$. The largest core number is called the \emph{degeneracy}
of $G$, denoted $\alpha(G)$.

The \emph{degeneracy DAG} of $G$,
denoted $D(G)$ is obtained by orienting edges in degeneracy order.
In other words, every edge $(u,v) \in G$ is directed from
lower to higher in the degeneracy ordering.
\end{definition}

The degeneracy ordering is the deletion time of the standard
linear time procedure that computes the degeneracy ~\cite{MB83}.
It is convenient for us to think of the degeneracy in terms of graph
orientations. As defined earlier, any permutation on $V$ can be used to make a DAG out of $G$.
We use this idea for generating saturated clique shadows.
Essentially, while $G$ may be sparse, \emph{out-neighborhoods} in $G$
are typically dense. (This has been observed in numerous results
on dense subgraph discovery~\cite{AnCh09,Tsourakakis13,SaSePi14}.)

We now define the procedure \shadow$(G,k)$, which works
by a simple, iterative refinement procedure. Think of $\bT$
as the current working set, and $\bS$ as the final output.
We take a set $(S,\ell)$ in $\bT$, and construct all
outneighborhoods in the degeneracy DAG. Any such
set whose density is above the Tur\'{a}n threshold goes to $\bS$ (the output),
otherwise, it goes to $\bT$ (back to the working set).

\begin{algorithm}
\caption{\shadow$(G,k)$}
 Initialize $\bT = \{(V,k)\}$ and $\bS = \emptyset$\;
 While $\exists (S,\ell) \in \bT$ such that $\rho_2(S) \leq 1 - \frac{1}{\ell-1}$\;
 \ \ \ \ Construct the degeneracy DAG $D(G|_S)$\;
 \ \ \ \ Let $N^+_s$ denote the outneighborhood (within $D(G|_S)$) of $s \in S$\;
 \ \ \ \ Delete $(S,\ell)$ from $\bT$\;
 \ \ \ \ For each $s \in S$\;
 \ \ \ \ \ \ \ If $\ell \leq 2$ or $\rho_2(N^+_s) > 1 - \frac{1}{\ell-2}$\;
 \ \ \ \ \ \ \ \ \ \label{step:add} Add $(N^+_s,\ell-1)$ to $\bS$\; 
 \ \ \ \ \ \ \ \label{step:move} Else, add $(N^+_s,\ell-1)$ to $\bT$\; 
 Output $\bS$\;
 \end{algorithm}

It is useful to define the \emph{recursion tree} $\cT$ of this process as
follows. Every pair $(S,\ell)$ that is ever part of $\bT$ is a node in $\cT$.
The children of $(S,\ell)$ are precisely the pairs $(N^+_s,\ell-1)$ added
in \Step{add}. (At the point, $(S,\ell)$ is deleted from $\bT$, and all
the $(N^+_s,\ell-1)$ are added.) Observe that the root of $\bT$ is $(V,k)$,
and the leaves are precisely the final output $\bS$.

\begin{theorem} \label{thm:shadow-output} The output $\bS$ of \shadow$(G,k)$
is a $\gamma$-saturated $k$-clique shadow, where $\gamma = 1/\max_{(S,\ell) \in \bS} (f(\ell) |S|^2)$.
\end{theorem}

\begin{proof} We first 
prove by induction the following loop invariant for \shadow: $\bT \cup \bS$ is always
a $k$-clique shadow. For the base case, note that 
at the beginning, $\bT = \{(V,k)\}$ and $\bS = \emptyset$. 
For the induction step, assume that $\bT \cup \bS$ is a $k$-clique
shadow at the beginning of some iteration. The element $(S,\ell)$
is deleted from $\bT$. Each $(N^+_s,\ell-1)$ is added to $\bS$
or to $\bT$. 

Thus, it suffices to prove that there is a bijection mapping between
$C_\ell(S)$ and $\bigcup_{s \in S} C_{\ell-1}(N^+_s)$. (By the induction
hypothesis, we can then construct a bijection between $C_k(G)$ and the appropriate
cliques in $\bT \cup \bS$.) Consider an $\ell$-clique $K$ in $S$.
Set $s$ to be the minimum vertex according to the degeneracy ordering
in $D(G|_S)$. Observe that the remaining vertices form an $(\ell-1)$-clique in $N^+_s$,
which we map the $K$ to. This is a bijection, because every clique $K$ can be mapped
to a (unique) $(\ell-1)$-clique, and furthermore, every $(\ell-1)$-clique in $\bigcup_{s \in S}
C_{\ell-1}(N^+_s)$ is in the image of this mapping.

 Thus, when \shadow{} terminates, $\bT \cup \bS$ is a $k$-clique shadow.
 Since $\bT$ must be empty, $\bS$ is a $k$-clique shadow. Furthermore,
 a pair $(S,\ell)$ is in $\bS$ iff $\rho_2(S) > 1 - \frac{1}{\ell-1}$.
 By \Thm{sat-shadow}, $\bS$ is $1/\max_{(S,\ell) \in \bS} (f(\ell)|S|^2)$-saturated.
 \end{proof}

We have a simple, but important claim that bounds the size
of any set in the shadow by the degeneracy.

 \begin{claim} \label{clm:degen} Consider non-root $(S,\ell) \in \cT$.
 Then $|S| \leq \alpha(G)$.
 \end{claim}

\begin{proof} Suppose the parent of $(S,\ell)$ is $(P,\ell+1)$. 
Observe that $S$ is the outneighborhood of some node $p$
in the DAG $D(G|_P)$. Thus, $|S| \leq \alpha(G|_P)$. The 
degeneracy can never be larger in a subgraph. (This is apparent
by an alternate definition of degeneracy, the maximum smallest degree of an induced
subgraph~\cite{MB83}.) Hence, $\alpha(G|_P) \leq \alpha(G)$.
\end{proof}

 \begin{theorem} \label{thm:shadow-time} The running time of \shadow$(G,k)$
 is $O(\alpha(G)\sz{\bS}+m+n)$. The total storage is $O(\sz{\bS}+m+n)$.
 \end{theorem} 

 \begin{proof} Every time we add $(N^+_S,\ell-1)$ (\Step{add}) to $\bT$,
 we explicitly construct the graph $G|_{N^+_s}$. Thus, we can guarantee
 that for every $(S,\ell)$ present in $\bT$, we can make queries in the graph $G|_S$.
 This construction takes $O(|S|^2)$ time, to query every pair in $S$. (This 
 is \emph{not} required when $S = V$, since $G|_V = G$.) Furthermore,
 this construction is done for every $(S,\ell) \in \cT$, except for the root node in $\cT$.
 Once we have $G|_S$, the degeneracy order can be computed in time
 linear in the number of edges in $G|_S$ ~\cite{MB83}. 

 Thus, the running time can be bounded by $O(\sum_{(S,\ell) \in \cT: S \neq V} |S|^2 + m + n)$.
 By \Clm{degen}, we can bound $\sum_{(S,\ell) \in \cT: S \neq V} |S|^2 = O(\alpha(G) \sum_{(S,\ell) \in \cT} |S|)$. 
 We split the sum over leaves and non-leaves. The sum over leaves is precisely
 a sum over the sets in $\bS$, so that yields $O(\alpha(G)\sz{\bS})$.
 It suffices to prove that $\sum_{(S,\ell) \in \cT: S \ \textrm{non-leaf}} |S| = O(\sz{\bS})$,
 which we show next.

 Observe that a non-leaf node $(S,\ell)$ in $\cT$ has exactly $|S|$ children, one for each
 vertex $s \in S$. Thus, 
 \begin{eqnarray*}
 \sum_{(S,\ell) \in \cT: (S,\ell) \textrm{non-leaf}} |S| & = & \sum_{(S,\ell) \in \cT} 
 \textrm{\# children of $(S,\ell)$} \\
 & = & \textrm{\# edges in $\cT$} 
 \end{eqnarray*}
 All internal nodes in $\cT$ have at least $2$ children, so the number of edges in $\cT$
 is at most twice the number of leaves in $\cT$. But this is exactly the number
 of sets in the output $\cS$, which is at most $\sz{\bS}$. 

 The total storage is $O(\sum_{(S,\ell) \in \cT} |S| + m + n)$, which
 is $O(\sz{S} + m + n)$ by the above arguments.
 \end{proof}

We now formally define the Tur\'{a}n shadow to be output of this procedure.

\begin{definition} \label{def:turan-shadow} The $k$-clique Tur\'{a}n shadow
of $G$ is the output of \shadow$(G,k)$.
\end{definition}

 \subsection{Putting it all together} \label{sec:together}

 \begin{algorithm}
 \caption{\mainalg$(G,k,\eps,\delta)$}
 Compute $\bS = $ \shadow$(G,k)$\;
 \label{step:gamma} Set $\gamma = 1/\max_{(S,\ell) \in \bS} (f(\ell) |S|^2)$\;
 Output $\hat{C}_k = \sample(G,k,\gamma,\eps,\delta)$\;
 \end{algorithm}

 \begin{theorem} \label{thm:main-full} Consider graph $G=(V,E)$ with $m$ edges,
 $n$ vertices, and degeneracy $\alpha(G)$. Assume $m \leq n^2/4$. 
 Let $\bS$ be the Tur\'{a}n $k$-clique shadow of $G$.

 With probability at least $1-\delta$
 (this probability is over the randomness of \mainalg; there
 is no stochastic assumption on $G$), $|\hat{C}_k - |C_k(G)|| \leq \eps |C_k(G)|$.

 The running time of \mainalg{} is $O(\alpha(G)\sz{\bS} + f(k)m\log(1/\delta)/\eps^2 + n)$ and
 the total storage is $O(\sz{\bS} + m + n)$.
 \end{theorem}

\begin{proof} By \Thm{shadow-output}, $\bS$ is $\gamma$-saturated,
for $\gamma = 1/\max_{(S,\ell) \in \bS} f(\ell)|S|^2$.
Since $m \leq n^2/4$, the procedure \shadow$(G,k)$ cannot
just output $\{(V,k)\}$. All leaves in the recursion tree must have depth at least $2$,
and by \Clm{degen}, for all $(S,\ell) \in \bS$, $|S| \leq \alpha(G)$.
A classic bound on the degeneracy asserts that $\alpha(G) \leq \sqrt{2m}$ (Lemma 1 of~\cite{ChNi85}).
Since $f(\ell)$ is increasing in $\ell$, 
$\max_{(S,\ell) \in \bS} f(\ell) |S|^2 \leq 2f(k)m$. Thus, $\gamma = \Omega(1/(f(k)m))$.

By \Thm{sample}, the running time of \sample{} is $O(\sz{\bS} + \log(1/\delta)/(\gamma\eps^2))$,
which is $O(\sz{\bS} + f(k)m\log(1/\delta)/\eps^2)$.
\Thm{sample} also asserts the accuracy of the output.
Adding the bounds of \Thm{shadow-time}, we prove the running time and storage bounds.
\end{proof}

\subsection{The shadow size} \label{sec:shadow}

The practicality of \mainalg{} hinges on $\sz{\bS}$ being small.
It is not hard to prove a worst-case bound, using the degeneracy.

\begin{claim} \label{clm:size} $\sz{\bS} = O(n\alpha(G)^{k-2})$.
\end{claim}

\begin{proof} By arguments in the proof of \Thm{shadow-time},
we can show that $\sz{\bS}$ is at most the number of edges in $\cT$.
In $\cT$, the degree of the root is $n$, and by \Clm{degen}, 
the degree of all other nodes is at most $\alpha(G)$. The depth of the 
tree is at most $k-1$, since the value of $\ell$ decreases every step
down the tree. That proves that $n\alpha^{k-2}$ bound.
\end{proof}

This bound is not that interesting, and the Chiba-Nishizeki 
algorithm for exact clique enumeration matches this bound~\cite{ChNi85}.
Indeed, we can design instances where \Clm{size} is tight
(a set of $n/\alpha$ Erd\H{o}s-R\'{e}nyi graphs $G_{\alpha,1/3}$).
In any case, beating an exponential dependence on $k$ for any algorithm
is unlikely~\cite{CHKX04}.

\emph{The key empirical insight of this paper is that Tur\'{a}n clique shadows
are small for real-world graphs.} We explain in more detail in the next section;
\Fig{shadow_size} shows that the shadow sizes are typically less than $m$, and never more than
$10m$.

\section{Experimental results}\label{sec:results}

\textbf{Preliminaries:} We implemented our algorithms in {\tt C++} and ran our experiments on a
commodity machine equipped with a 3.00GHz Intel Core i7 processor with 8~cores
and  256KB  L2 cache (per core), 20MB L3 cache, and  128GB memory. 

We performed our experiments on a collection of graphs  from
SNAP~\cite{SNAP}, the largest with more than 100M edges. 
The collection includes social networks, web networks, and infrastructure networks.
Each graph is made simple by ignoring direction.
Basic properties of these graphs are presented in \Tab{estimates_TS}. 

In the implementation of \mainalg, there is just one parameter to choose:
the number of samples chosen in \Step{sample} in \sample. Theoretically,
it is set to $(20/\gamma\eps^2) \log(1/\delta)$; in practice, we 
just set it to 50K for all our runs. Note that $\gamma$ is not a free parameter
and is automatically set in \Step{gamma} of \mainalg.

We focus on counting $k$-cliques for $k$ ranging from $5$ to $10$.
We ignore $k=3,4$, since there is much existing (scalable) work
for this setting~\cite{SePiKo13,JhSePi15,AhNe+15}.
For the sake of presentation, we showcase results for $k=7,10$.
We focus on $k=10$ since no existing algorithm produces results for 10-cliques
in reasonable time. We also show specifics for $k=7$, to contrast
with $k=10$.


\begin{table*}[]
\centering
\caption{Graph properties}
\label{tab:estimates_TS}
\begin{adjustbox}{max width=\textwidth}
\begin{tabular}{|l|l|l|r|r|l|r|r|l|r|r|l|r|r|}
\hline
\multicolumn{5}{|c|}{}  & \multicolumn{3}{c|}{\textbf{k=5}}                       & \multicolumn{3}{c|}{\textbf{k=7}}                       & \multicolumn{3}{c|}{\textbf{k=10}}                      \\ \cline{7-7}
\hline
\textbf{graph}         & \textbf{vertices}     & \textbf{edges}        & \textbf{degen}        & \textbf{max degree}   & \textbf{estimate} & \textbf{\% error} & \textbf{time} & \textbf{estimate} & \textbf{\% error} & \textbf{time} & \textbf{estimate} & \textbf{\% error} & \textbf{time} \\ \cline{7-7}
\hline
loc-gowalla            & 1.97E+05              & 9.50E+05              & 51                    & 14730                 & 1.46E+07          & 0.20              & 2             & 4.78E+07          & 0.36              & 2             & 1.08E+08          & 1.63              & 3             \\
web-Stanford           & 2.82E+05              & 1.99E+06              & 71                    & 38625                 & 6.21E+08          & 0.00                 & 20            & 3.47E+10          & 0.13                 & 43            & 5.82E+12          & -                 & 52            \\
amazon0601             & 4.03E+05              & 4.89E+06              & 10                    & 2752                  & 3.64E+06          & 0.93              & 1             & 9.98E+05          & 0.95              & 1             & 9.77E+03          & 0.01              & 1             \\
com-youtube            & 1.13E+06              & 2.99E+06              & 51                    & 28754                 & 7.29E+06          & 1.08              & 7             & 7.85E+06          & 1.38              & 8             & 1.83E+06          & 0.20              & 8             \\
web-Google             & 8.76E+05              & 4.32E+06              & 44                    & 6332                  & 1.05E+08          & 0.10              & 2             & 6.06E+08          & 0.09              & 2             & 1.29E+10          & 0.82              & 2             \\
web-BerkStan           & 6.85E+05              & 6.65E+06              & 201                   & 84230                 & 2.19E+10          & 0.00                 & 101           & 9.30E+12          & 1.05                 & 214           & 5.79E+16          & -                 & 262           \\
as-skitter             & 1.70E+06              & 1.11E+07              & 111                   & 35455                 & 1.17E+09          & 0.01                 & 153           & 7.30E+10          & 0.23                 & 164           & 1.42E+13          & -                 & 180           \\
cit-Patents            & 3.77E+06              & 1.65E+07              & 64                    & 793                   & 3.05E+06          & 0.34              & 10            & 1.89E+06          & 0.83              & 9             & 2.55E+03          & 4.46              & 9             \\
soc-pokec              & 1.63E+06              & 2.23E+07              & 47                    & 14854                 & 5.29E+07          & 0.13              & 42            & 8.43E+07          & 0.48              & 45            & 1.98E+08          & 0.01              & 45            \\
com-lj                 & 4.00E+06              & 3.47E+07              & 360                   & 14815                 & 2.46E+11          & -                 & 106           & 4.48E+14          & -                 & 153           & 1.47E+19          & -                 & 252           \\
com-orkut              & 3.07E+06              & 1.17E+08              & 253                   & 33313                 & 1.57E+10          & 0.00                 & 3119          & 3.61E+11          & 1.97                 & 5587          & 3.03E+13          & -                 & 9298       \\  
\hline
\end{tabular}
\end{adjustbox}
\caption{Table shows the sizes, degeneracy, maximum degree of the graphs, the counts of 5, 7 and 10 cliques obtained using \mainalg{}, the percent relative error in the estimates, and time in seconds required to get the estimates. Some of the exact counts were obtained from ~\cite{FFF15} (where available). This is the first such algorithm that obtains these counts with $<2\%$ error without using any specialized hardware.}
\end{table*}

\textbf{Convergence of \mainalg:}
We picked two smaller graphs {\tt amazon0601} and {\tt web-Google} for which
the exact $k$-clique count is known (for all $k \in [5,10]$). We choose both $k=7, 10$.
For each graph, for sample size in [10K,50K,100K,500K,1M], we perform 100 runs
of the algorithm. We plot the spread of the output of \mainalg{}, over
all these runs. The results are shown in~\Fig{convergence}. The red line
denotes the true answer, and there is a point for the output of every
single run. Even for 10-clique counting,
the spread of 100 runs is absolutely minimal.
For 50K samples, the range of values is within 2\% of the true answer. 
This was consistent with all our runs. 

\textbf{Accuracy of \mainalg:} 
For many graphs (and values of $k$), it was not feasible to get an exact algorithm
to run in reasonable time. The run time of exact procedures can vary
wildly, so we have exact numbers for some larger graphs but could not
generate numbers for smaller graphs. 
We collected as many exact results as possible to validate \mainalg.
For the sake of presentation, we only show a snapshot of these results here.

For $k=7$, we collected exact results for a collection of graphs,
and for each graph, compared the output of a single run of \mainalg{} (with 50K samples)
with the true answer. We compute \emph{relative error}: |true - estimate|/true. 
These results are presented in~\Fig{acc}. Note that the
errors are within 2\% in all cases, again consistent with all our runs.

In \Tab{estimates_TS}, we present the output of our algorithm for a single run
on all instances and $k=5,7,10$. For every graph where we know the true value,
we present the relative error. Barring one example ({\tt cit-Patents} for $k=10$),
all errors are less than 2\%. Even in the worst case, the error is at most 5\%.

\begin{figure*}
\begin{subfigure}[b]{0.24\textwidth}
    \centering
    \includegraphics[width=\textwidth]{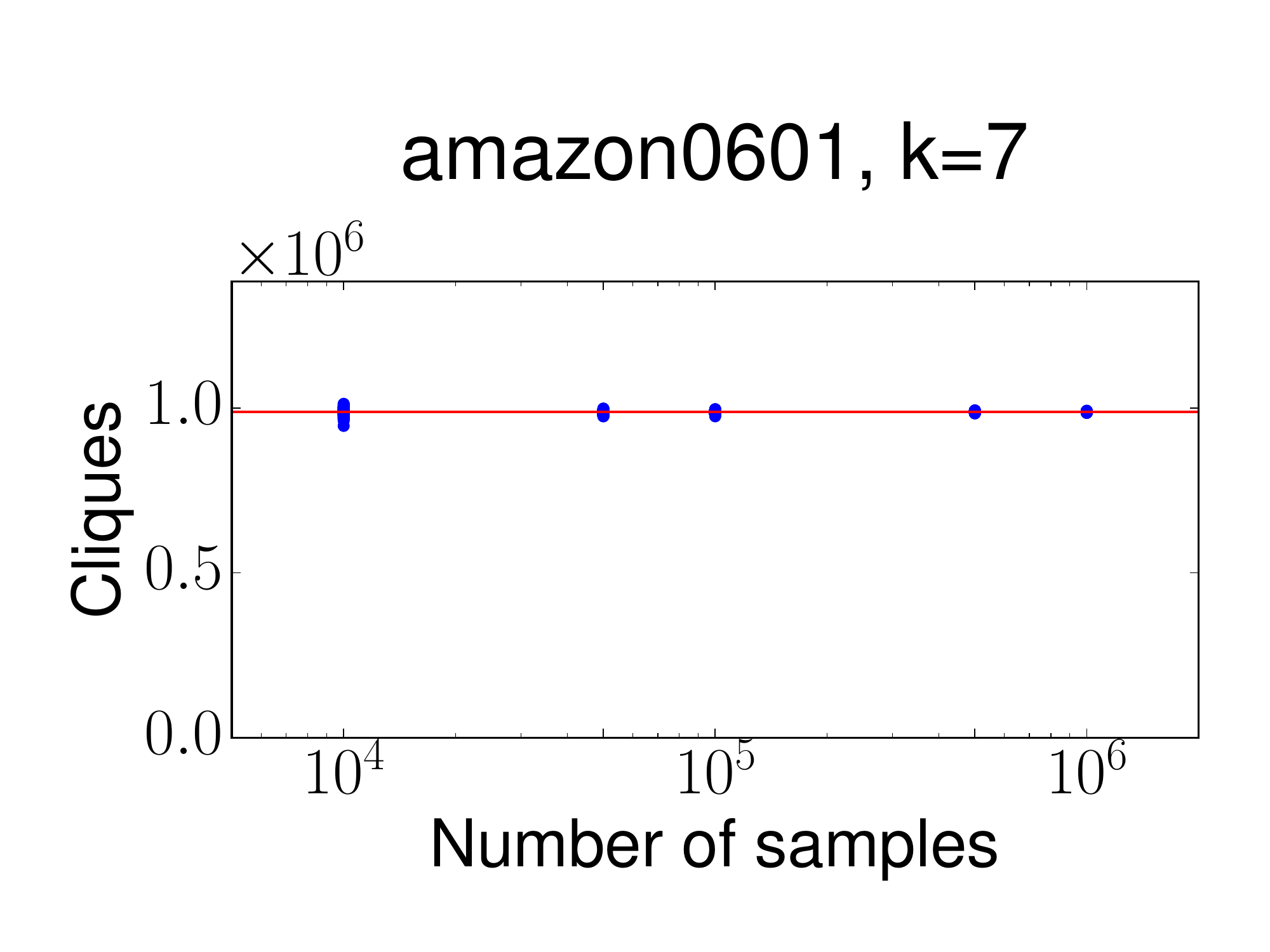}
    \label{fig:p1}
\end{subfigure}
\hfill
\begin{subfigure}[b]{0.24\textwidth}
    \centering
    \includegraphics[width=\textwidth]{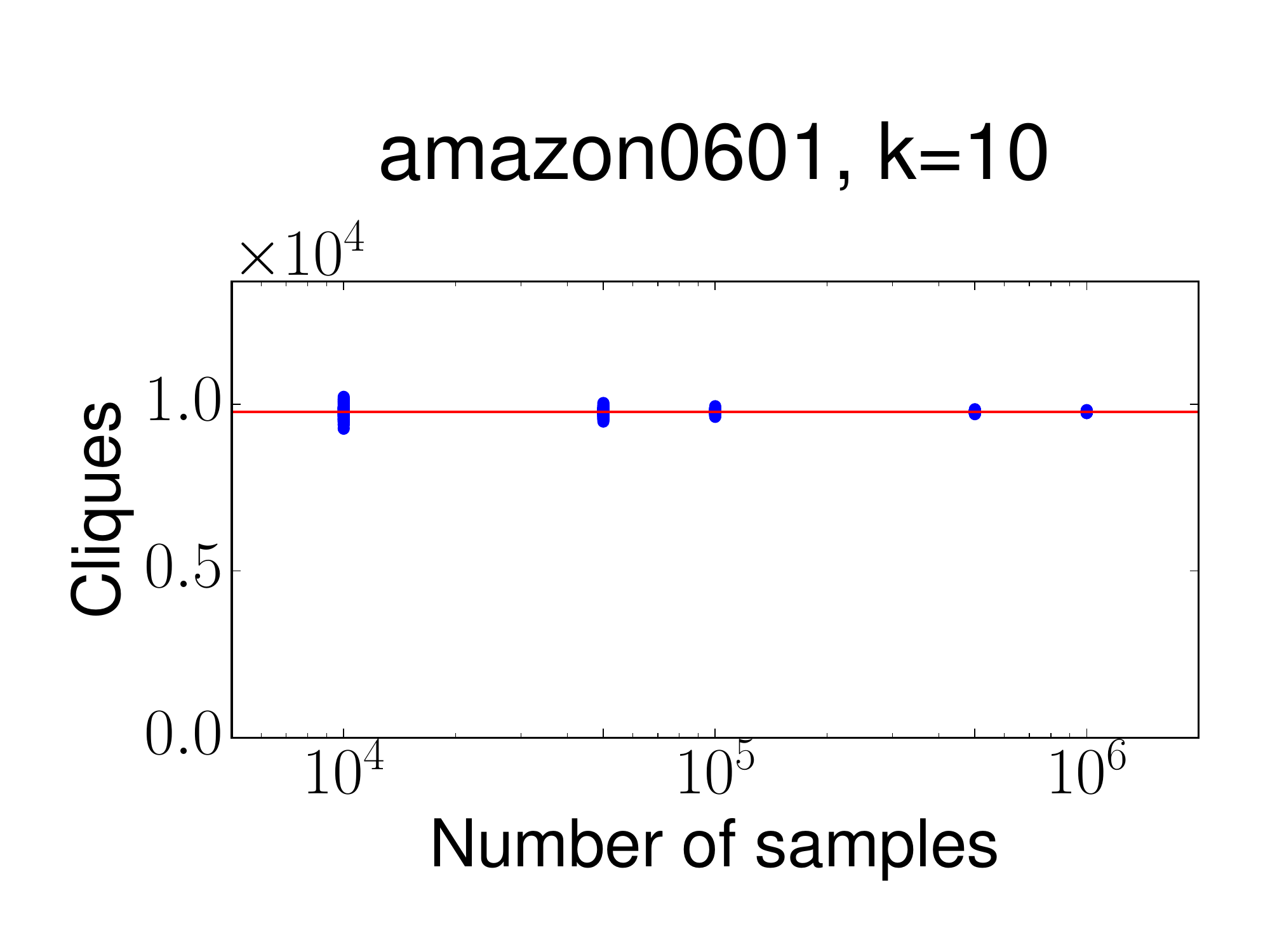}
    \label{fig:p2}
\end{subfigure}
\hfill
\begin{subfigure}[b]{0.24\textwidth}
    \centering
    \includegraphics[width=\textwidth]{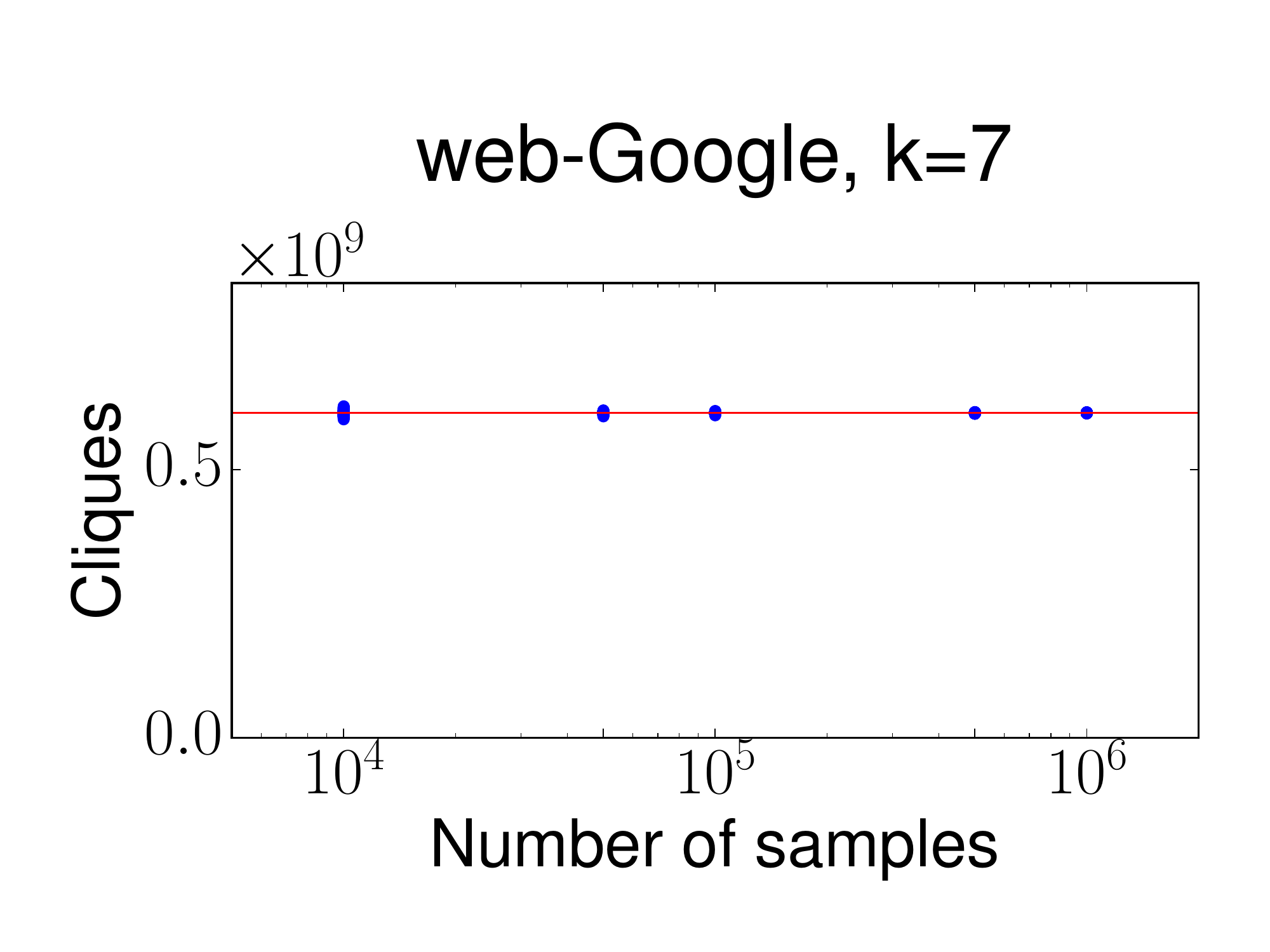}
    \label{fig:p3}
\end{subfigure}
\hfill
\begin{subfigure}[b]{0.24\textwidth}
    \centering
    \includegraphics[width=\textwidth]{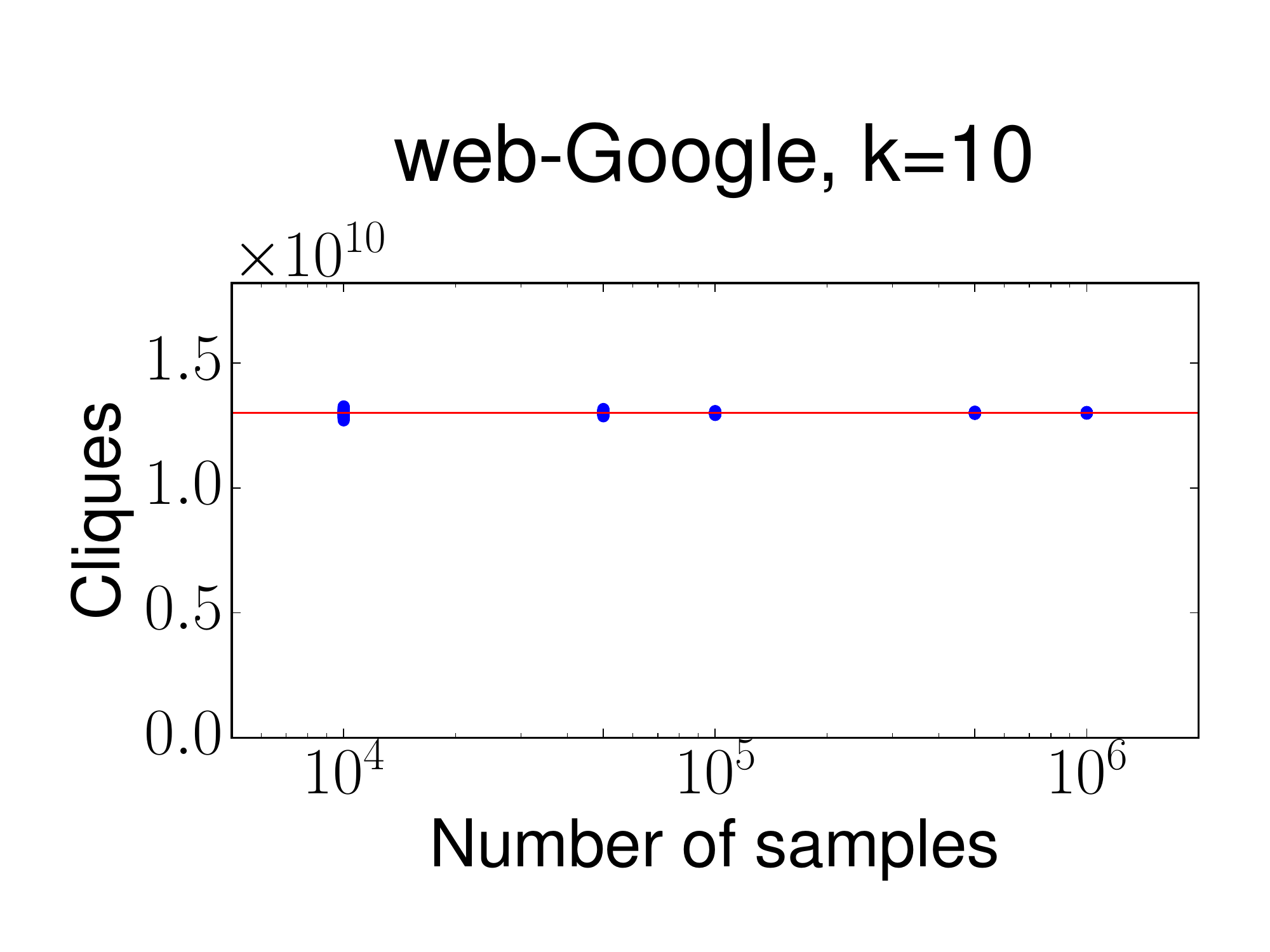}
    \label{fig:p4}
\end{subfigure}
\hfill
\caption{\footnotesize{ Figure shows convergence over 100 runs of \mainalg{} using 10K, 50K, 100K, 500K and 1M samples each. \mainalg{} has an extremely low spread and consistently gives very accurate results.}}
\label{fig:convergence}
\end{figure*}

%
%
%
%
%
%
%

\textbf{Running time:} All runtimes are presented in \Tab{estimates_TS}. (We show
the time for a single run, since there was little variance for different runs on
the same graph.) In all instances except {\tt com-orkut}, the runtime was a few minutes,
even for graphs with tens of millions of edges. We stress that these are all on a single
machine. For {\tt com-orkut}, the runtime is at most 2.5 hours. Previously, such
graphs were processed with MapReduce on clusters~\cite{FFF15}.

\begin{figure}[t]   
        \begin{subfigure}[b]{0.22\textwidth}
        \centering
        \includegraphics[width=\textwidth]{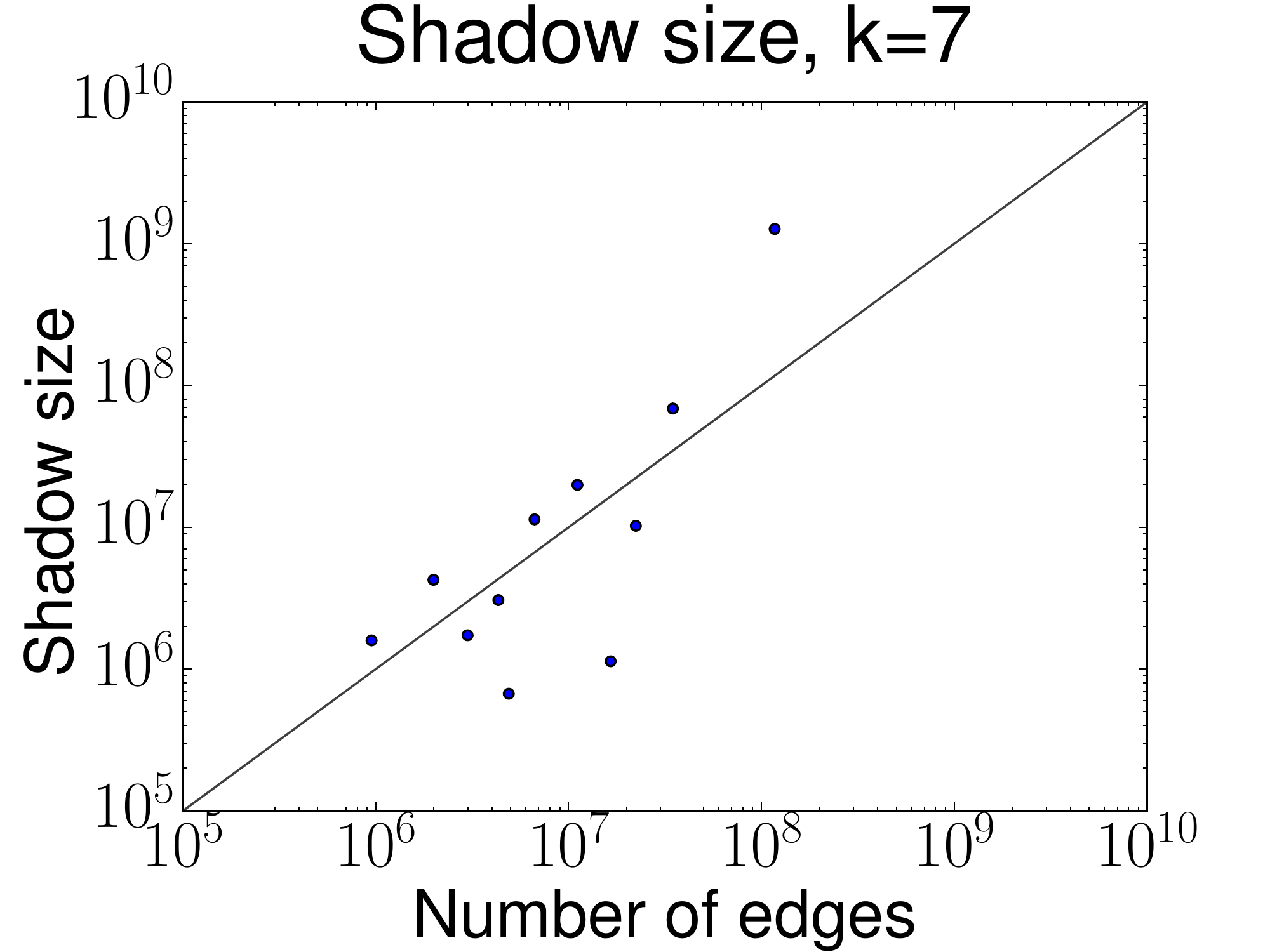}
        \label{fig:shadow7}
        \end{subfigure}
        \hfill
        \begin{subfigure}[b]{0.22\textwidth}
        \centering
        \includegraphics[width=\textwidth]{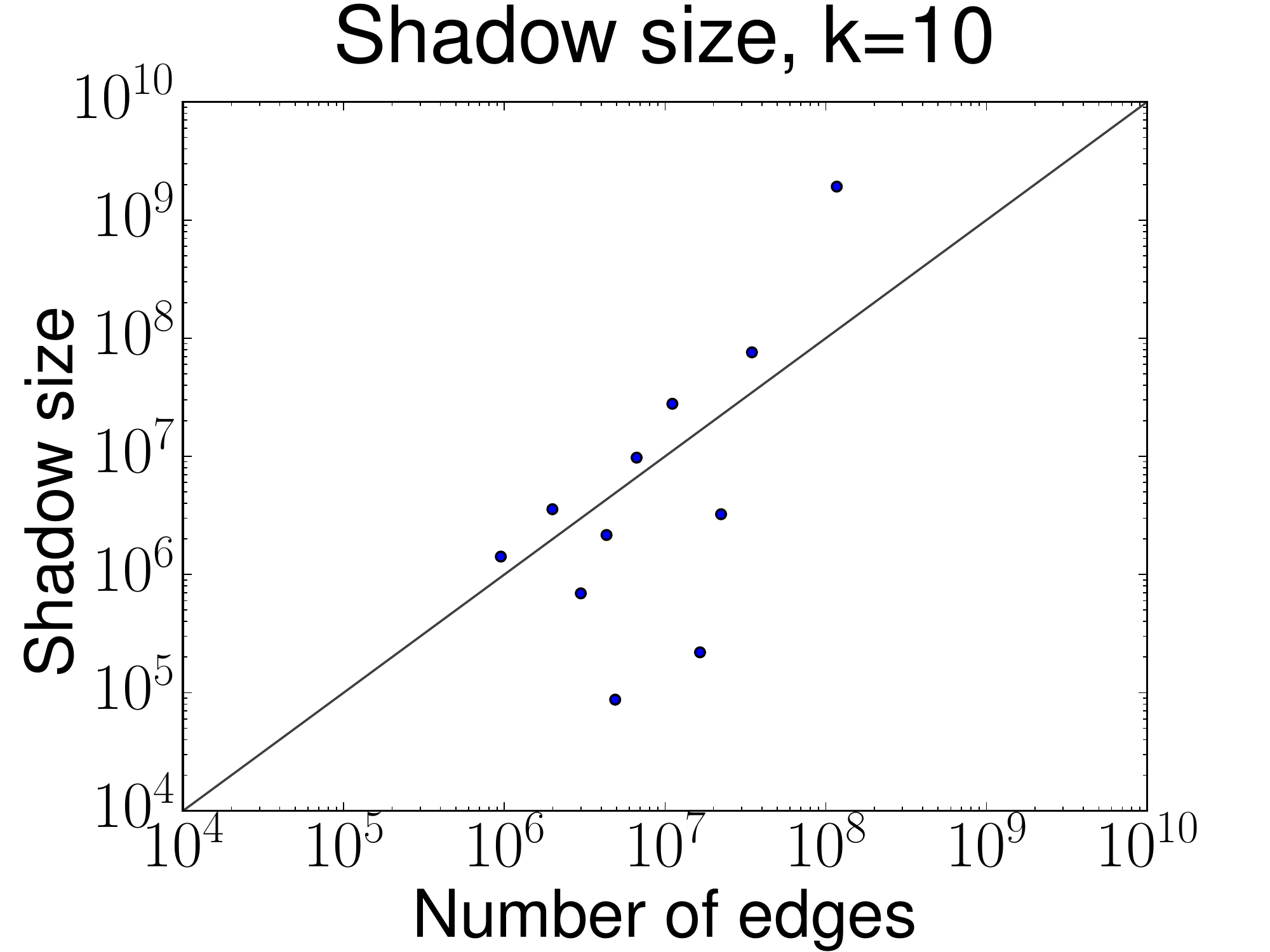}
        \label{fig:shadow10}
        \end{subfigure}
        \caption{\footnotesize{ Figures show the sizes of the Tur\'{a}n shadows generated for k=7 and k=10 in all the graphs. The runtime of the algorithm is proportional to the size of the shadow and crucially, the sizes scale only linearly with the number of edges.}}
        \label{fig:shadow_size}
\end{figure}

\begin{figure}[t]   
        \begin{subfigure}[b]{0.22\textwidth}
        \includegraphics[width=\textwidth]{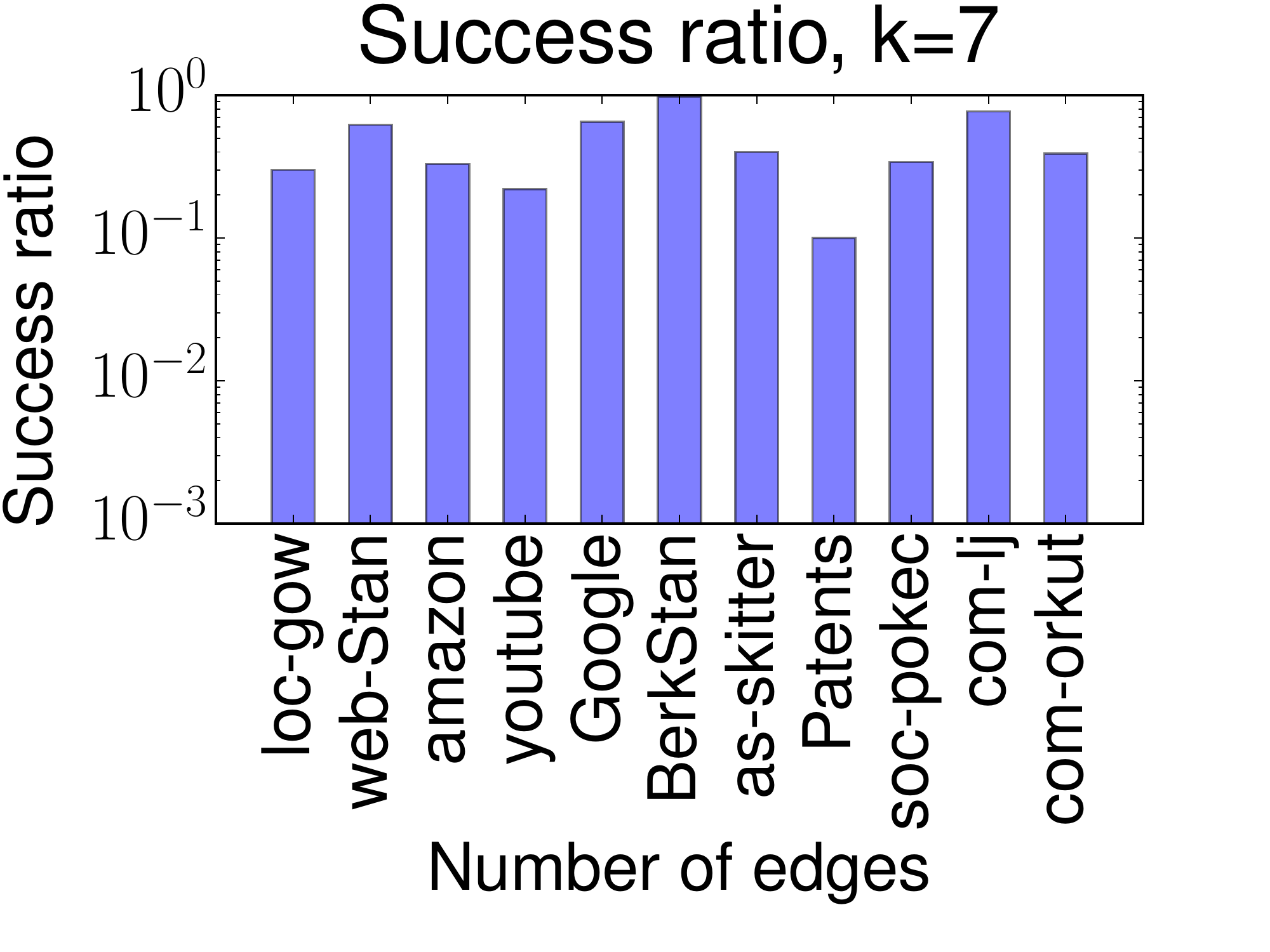}
        \label{fig:p5}
        \end{subfigure}
        \begin{subfigure}[b]{0.22\textwidth}
        \includegraphics[width=\textwidth]{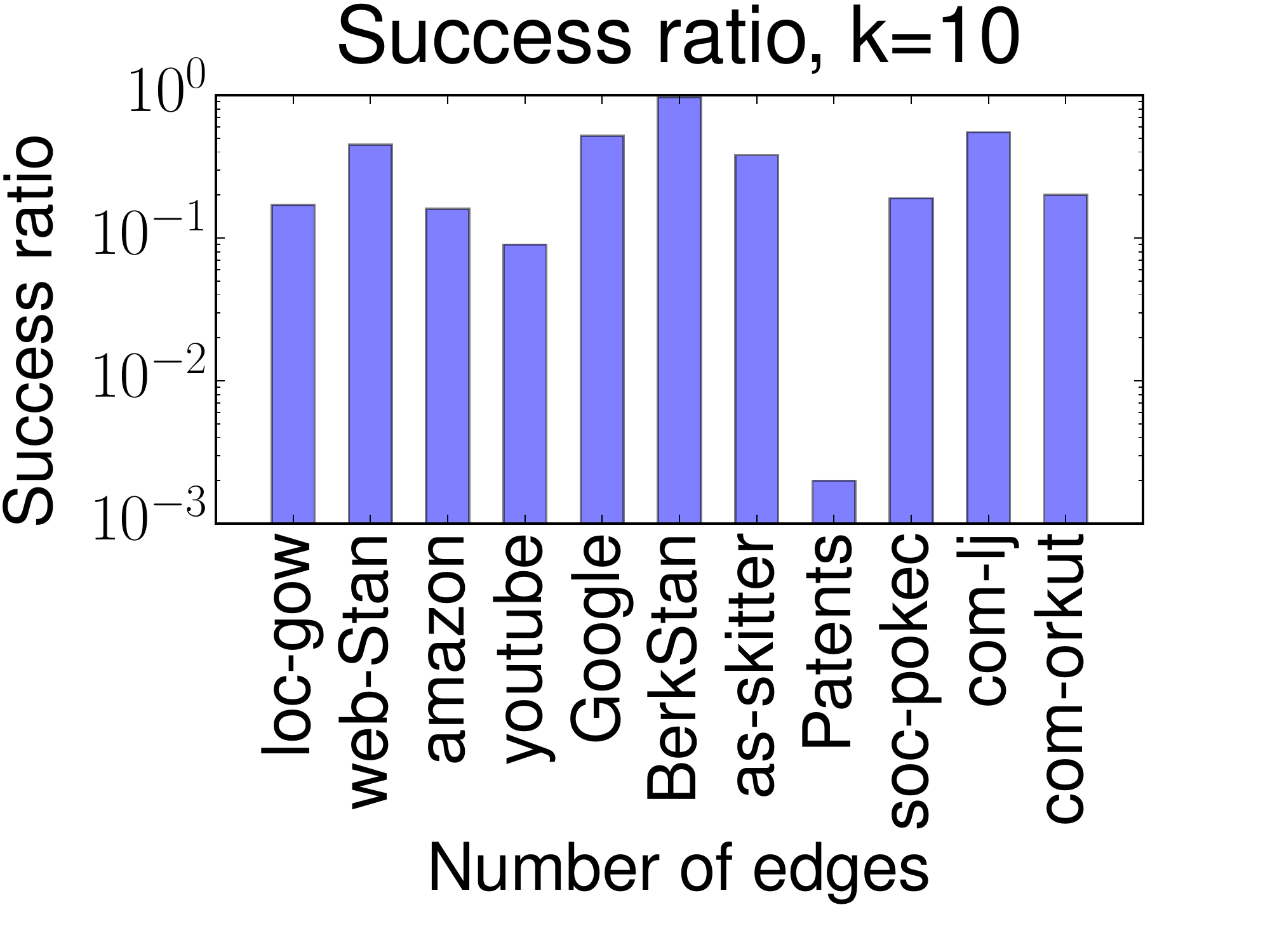}
        \label{fig:succ10}
        \end{subfigure}
        \caption{\footnotesize{ Figures show the success ratio (probability of finding a clique) obtained in the sampling experiments in all the graphs. }}
        \label{fig:success_ratios}
\end{figure}

\begin{figure}[t]   
        \begin{subfigure}[b]{0.45\textwidth}
        \centering
        \includegraphics[width=\textwidth]{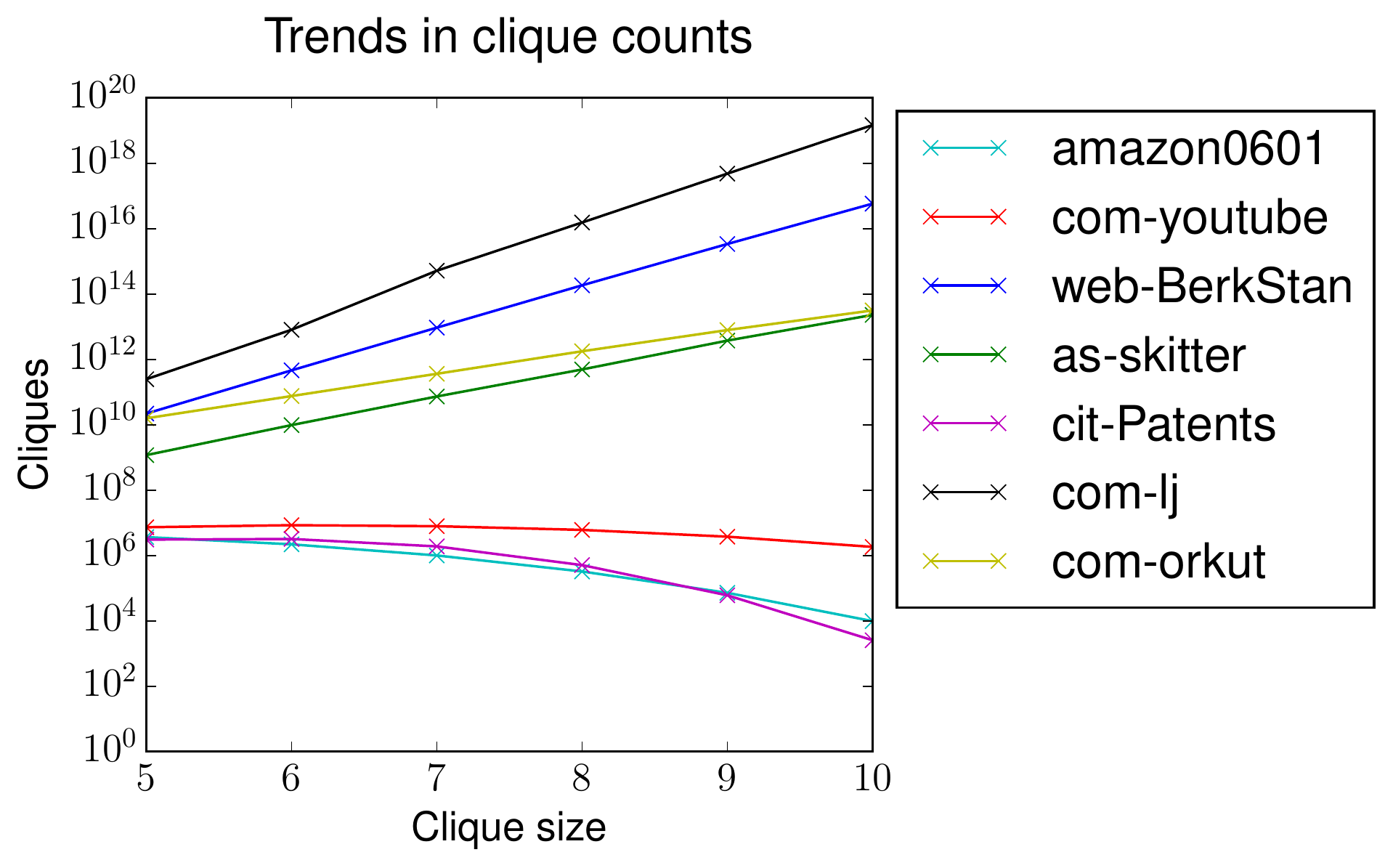}
        \end{subfigure}
        \caption{\footnotesize{ Figures show the trends in clique counts of some graphs. While cit-Patents, com-youtube and amazon0601 show a decreasing trend, all other graphs show an exponential increase in the number of cliques with clique size. }}
        \label{fig:trends}
\end{figure}

\subsection{Comparison with other algorithms}  
Our exact brute-force procedure is a
well-tuned algorithm that uses the degeneracy ordering and exhaustively searches outneighborhoods
for cliques. This is basically the procedure of Finetti {\em et al.}~\cite{FFF15}, inspired
by the algorithm of Chiba-Nishizeki~\cite{ChNi85}. 
We compare with the following algorithms.
\begin{asparaitem}
    \item Color coding: This is a classic algorithmic technique~\cite{AlYuZw94}. For counting $k$-cliques, the algorithm randomly colors vertices with one of $k$ colors. Then, the algorithm
    uses a brute-force procedure to count polychromatic $k$-cliques (where each vertex
    has a different color). This number is scaled to 
    give an unbiased estimate, and the coloring helps cut
    down the search time of the brute-force procedure. This method has been applied
    in practice for numerous pattern counting problems~\cite{HoBe+07,BetzlerBFKN11,ZhWaBu+12}. 
    \item Edge sampling: Edge sampling was discussed by Tsourakakis
    {\em et al.} in the context of triangle counting~\cite{TsKaMiFa09,TsDrMi09,TsKoMi11}, though the idea is
    flexible and can be used for large patterns~\cite{ElShBo16}. The idea here is to sample
    each edge independently with some probability $p$, and then count $k$-cliques
    in the down-sampled graph. This number is scaled to give an unbiased estimate 
    for the number of $k$-cliques.

    For clique counting, we observe that minor differences in $p$ (by $0.1$) 
    have huge effects on runtime and accuracy. To do a fair comparison, we run multiple experiments
    with varying $p$ (increments of $0.1$), until we reach the smallest $p$ that consistently
    yields less than 5\% error. (Note that the error of \mainalg{} is significantly smaller
    that this.) Timing comparisons are done with runs for that value of $p$.

    \item GRAFT~\cite{RaBhHa14}: Rahman {\em et al.} give a variant of edge sampling 
    with better performance for large pattern counts~\cite{RaBhHa14}.
    The idea is to sample some set of edges, and exactly count the number of $k$-cliques
    on each of these edges. This can be scaled into an unbiased estimate for the total
    number of $k$-cliques.

    As with edge sampling, we increase the number of edge samples until we get consistently
    within 5\% error. Timing comparisons are done with this setting. Typical settings
    seems to be in the range of 100K to 1M samples. Beyond that, GRAFT is infeasible, even
    for graphs with 10M edges. 
\end{asparaitem}

We focus on $k=7,10$ for clarity. In all cases, we simply terminate the algorithm
if it takes more than the minimum of 7 hours and 100 times the time required by \mainalg{}.
We present the speedup of \mainalg{} with respect to all these algorithms
in \Fig{7-speedup} for k=7. For k=10, for most instances, no competing algorithm terminated.
\begin{asparaitem}
    \item $k=7$ (\Fig{7-speedup}): \mainalg{} outperformed Color Coding and GRAFT across all instances. Color Coding never gave good accuracy, so we ignore it in our speedup plots.
    We do note that Edge Sampling gives extremely good performance in some instances,
    but can be very slow in others. For {\tt amazon0601}, {\tt com-youtube},
    {\tt cit-Patents}, and {\tt soc-pokec}, Edge Sampling is faster than \mainalg.
    But \mainalg{} handles all these graphs with a minute. The only exception is 
    {\tt com-orkut}, where GRAFT is much faster than \mainalg.
    We note that all other algorithms can perform extremely poorly on fairly small graphs:
    Edge Sampling is 10-100 times slower on a number of graphs, which have only millions
    of edges.  On the other hand, \mainalg{} always runs in minutes for these graphs.

    \item $k=10$ : \emph{No competing algorithm} is able to handle 10 cliques for all datasets,
    even in 7 hours (giving a speedup of anywhere between 3x to 100x). They all generally fail
    for at least half of the instances. \mainalg{}
    gets an answer for {\tt com-orkut} within 2.5 hours, and handles all other
    graphs in minutes. 
\end{asparaitem}

\subsection{Details about \mainalg}

\textbf{Shadow size:} In \Fig{shadow_size}, we plot the size of the $k$-clique Tur\'{a}n shadow
with respect to the number of edges in each instance. This is done for $k=7,10$.
(The line $y=x$ is drawn as well.) As seen from \Thm{main-full}, the size
of the shadow controls the storage and runtime of \mainalg.
We see how in almost all instances, the shadow
size is around the number of edges. This empirically explains the efficiency of \mainalg.
The worst case is {\tt com-orkut}, where the shadow size is at most ten
times the number of edges.

\textbf{Success probability:} The final estimate of \mainalg{} is generated through
\sample. We asserted (theoretically) that $O(m)$ samples suffice, and in practice,
we use 50K samples. In \Fig{success_ratios}, we plot (for $k=7,10$) the empirical
probability of finding a clique in \Step{X} of \sample. The higher this is,
the fewer samples we require and the more confidence in the statistical validity of our estimate. Almost all (empirical) probabilities are more than
$0.1$, and 50K samples are more than enough for convergence.

\textbf{Trends in clique numbers:} \Fig{trends} plots the number of $k$-cliques
(as computed by \mainalg) versus $k$. (We do not consider all graphs for the sake of clarity.)
Interestingly, there are some graphs where the number of cliques grows exponentially.
This is probably because of a large clique/dense-subgraph, and it would be interesting
to verify this. For another class of graphs, the clique counts are consistently
decreasing. This seems to classify graphs into one of two types. We feel further analysis
of these trends would be interesting, and \mainalg{} can be a useful tool
for network analysis.

%
%

\clearpage

\scriptsize

\bibliographystyle{abbrv}
\bibliography{cliquesampling}

\end{document}